\documentclass[a4paper,USenglish,11pt]{article}

\usepackage[margin=1in]{geometry}
  \usepackage{hyperref}
  \hypersetup{colorlinks=true,citecolor=blue}
  \usepackage{amsthm,amssymb,amsmath}
  \usepackage{xspace}
  \usepackage{makecell}
  \usepackage[ruled,noend,linesnumbered]{algorithm2e}
  \usepackage{tabularx}
  \usepackage{comment}
  \usepackage{todonotes}
  \usepackage{thmtools}
  \usepackage{thm-restate}
  \usepackage{authblk}
  \usepackage[capitalise]{cleveref}
  \usepackage{verbatim}
  \usepackage{multicol}
  \usepackage[utf8]{inputenc}
  \usepackage[OT4]{fontenc}
  \usepackage[noadjust]{cite}
  \usepackage{enumitem}
  \usetikzlibrary{patterns}

\DeclareMathOperator{\polylog}{polylog}
\DeclareMathOperator{\polyloglog}{polyloglog}

\usetikzlibrary{trees, arrows, shapes, snakes}

\tikzset{
  treenode/.style = {align=center, inner sep=2pt, text centered,
    font=\sffamily},
  arn_r/.style = {treenode, circle, black, font=\sffamily\bfseries, draw=black,
    text width=1.5em},
    arn_t/.style = {treenode, circle, black, thick, double, font=\sffamily\bfseries, draw=black,
    text width=1.5em},
  every edge/.append style={anchor=south,auto=falseanchor=south,auto=false,font=3.5 em},
}

\def\dd{\mathinner{.\,.}}

\newcommand{\cO}{\mathcal{O}}
\newcommand{\Oh}{\cO}
\newcommand{\cOtilde}{\tilde{\mathcal{O}}}
\newcommand{\Ohtilde}{\cOtilde}

\newcommand{\LCA}{\textsf{LCA}}

\newcommand{\weight}{\operatorname{weight}}

\newcommand{\pre}{\operatorname{pre}}
\newcommand{\Pred}{\textsc{PredecessorQuery}\xspace}
\newcommand{\II}{\mathcal{I}}

\newcommand{\enumcases}[1]{
\begin{enumerate}[label=Case \arabic*:,leftmargin=42pt]
 #1
\end{enumerate}
}

 \newcommand{\defDSproblem}[3]{
  \vspace{2mm}
\noindent\fbox{
  \begin{minipage}{0.96\textwidth}
  #1\\
  {\bf{Input:}} #2  \\
  {\bf{Query:}} #3
  \end{minipage}
  }
  \vspace{2mm}
}

  \newtheorem{theorem}{Theorem}
  \newtheorem{lemma}[theorem]{Lemma}
  \newtheorem{corollary}[theorem]{Corollary}
  \newtheorem{proposition}[theorem]{Proposition}

  \newtheorem{definition}[theorem]{Definition}
  
  \newtheorem{property}[theorem]{Property}

\begin{document}

\title{Optimal Heaviest Induced Ancestors}

\author[1]{Panagiotis Charalampopoulos}
\author[2]{Bart{\l}omiej Dudek}
\author[2]{Pawe{\l} Gawrychowski}
\author[2]{Karol Pokorski}

\affil[1]{Birkbeck, University of London, UK\\}
\affil[2]{Institute of Computer Science, University of Wroc{\l}aw, Poland}
\date{}
\maketitle

\begin{abstract}
We revisit the Heaviest Induced Ancestors (HIA) problem that was introduced by Gagie, Gawrychowski, and Nekrich [CCCG 2013] and has a number of applications in string algorithms.
Let $T_1$ and $T_2$ be two rooted trees whose nodes have weights that are increasing in all root-to-leaf paths, and labels on the leaves, such that no two leaves of a tree have the same label.
A pair of nodes $(u, v)\in T_1 \times T_2$ is \emph{induced} if and only if there is a label shared by leaf-descendants of $u$ and $v$.
In an HIA query, given nodes $x \in T_1$ and $y \in T_2$,
the goal is to find an induced pair of nodes $(u, v)$ of the maximum total weight such that $u$ is an ancestor of~$x$
and $v$ is an ancestor of $y$.

Let $n$ be the upper bound on the sizes of the two trees.
It is known that no data structure of size $\cOtilde(n)$ can answer HIA queries in
$o(\log n / \log \log n)$ time [Charalampopoulos, Gawrychowski, Pokorski; ICALP 2020].\footnote{The $\cOtilde(\cdot)$ notation hides factors polylogarithmic in $n$.}
This (unconditional) lower bound is a $\polyloglog n$ factor away from the query time of the fastest $\cOtilde(n)$-size data structure known to date for the HIA problem [Abedin, Hooshmand, Ganguly, Thankachan; Algorithmica 2022].
In this work, we resolve the query-time complexity of the HIA problem for the near-linear space regime by presenting
a data structure that can be built in  $\cOtilde(n)$ time and answers HIA queries in $\cO(\log n/\log\log n)$ time.
As a direct corollary, we obtain an $\cOtilde(n)$-size data structure that maintains the LCS of a static string
and a dynamic string, both of length at most $n$, in time optimal for this space regime.

The main ingredients of our approach are fractional cascading and the utilization of an $\cO(\log n/ \log\log n)$-depth tree decomposition.
The latter allows us to break through the $\Omega(\log n)$ barrier faced by previous works, due to the depth of the considered heavy-path decompositions.
\end{abstract}

\section{Introduction}

The solutions to algorithmic problems on texts frequently involve the construction of text indexes
that can be built efficiently and offer a broad functionality, without significantly increasing space usage.
A prime example of such an index is the suffix tree, which is ubiquitous in stringology.
The work of Weiner \cite{Weiner73} that introduced it, showed that it can be used to efficiently solve a number of fundamental open problems such as the computation of occurrences of patterns (given in an online manner) in a text or the computation of the longest common substring of two strings.
However, it is usually the case that a suffix tree needs to first be augmented with other data structures before it can efficiently answer more sophisticated queries, e.g., returning the longest common prefix of two substrings or the longest palindrome centered at some position; an augmentation with a lowest common ancestors data structure suffices for these examples \cite{gusfield1997,HarelT84}.

Crucially, a text index, such as the suffix tree, is built once and can then be queried an arbitrary number of times.
This is increasingly relevant: in many real-world scenarios, large pieces
of information are stored on servers and are constantly queried by a large number of remote clients.
From this perspective, it makes sense to devote some time to preprocess the data stored on the server in order
to be able to provide quick responses to remote users later.

The \emph{Heaviest Induced Ancestors} problem, which was introduced by
Gagie et al.~\cite{DBLP:conf/cccg/GagieGN13} and is defined next, has been
proved to be useful in solving several variants of the problem of computing a longest common
substring of two strings~\cite{DBLP:conf/cccg/GagieGN13,Amir2017,DBLP:journals/algorithmica/AmirCPR20,DBLP:conf/cpm/AbedinH0T18,charalampopoulos_et_al}.

We say that a tree is weighted if there is a weight associated with each
node $u$ of the tree, such that weights along root-to-leaf paths are
increasing, i.e., for any node $u$ other than the root the weight of $u$
is larger than the weight of $u$'s parent.
Further, we say that a tree is labelled if each of its leaves is given a distinct
label from $[n]$, where $n$ is the number of leaves.
As an example of a rooted, weighted, and labelled tree, consider the suffix tree of a string $S\$$, where
$\$$ does not have any occurrence in $S$, with the label of each leaf being the starting position of the corresponding suffix
and the weight of each node being the length of the string it represents.

\begin{definition}
  For two rooted and weighted trees $T_1$ and $T_2$ on $n$ leaves,
  we say that two nodes $u \in T_1$ and $v \in T_2$,
  are \emph{induced} (by label $\ell$) if and only if there are
  leaves $x$ and $y$ labelled with $\ell$, such that $x$ and $y$ are weak descendants of $u$ and $v$, respectively.
\end{definition}

\defDSproblem{\textsc{Heaviest Induced Ancestors} (\textsc{HIA})}
{Two rooted, weighted, and labelled trees $T_1$ and $T_2$ on $n$ leaves.}
{Given a pair of nodes $u \in T_1$ and $v \in T_2$, return a pair of induced nodes $(u', v')$ with the largest total weight, such that $u'$ is an ancestor of $u$, $v'$ is an ancestor of $v$.}

\subparagraph{Previous results and our contribution.}
Table~\ref{table:results} shows the state-of-the-art size vs.~query-time tradeoffs for the HIA problem prior
to our work and our result.
Gagie et al.~\cite{DBLP:conf/cccg/GagieGN13} presented several tradeoffs which have been since improved.
We stress that the $\cO(n\log^2 n)$-size data structure with query-time $\cO(\log n)$ included in Table \ref{table:results} was only sketched in~\cite{DBLP:conf/cccg/GagieGN13}. We briefly discuss this sketch in \cref{app:fc}, as some of the ideas involved are similar to the ones we use.
The remaining $\cOtilde(n)$-size known data structures found in Table~\ref{table:results} are due to Abedin et al.~\cite{DBLP:conf/cpm/AbedinH0T18}.
Charalampopoulos et al.~\cite{charalampopoulos_et_al} showed an unconditional lower bound for near-linear size data structures and a data structure with query-time $\cO(1)$ and size $\cO(n^{1+\epsilon})$ for any constant $\epsilon>0$.
We now formally state our main result, which matches the lower bound of \cite{charalampopoulos_et_al}.

\begin{theorem}\label{thm:final}
  There is an $\cOtilde(n)$-size data structure for the HIA problem that can be constructed in $\cOtilde(n)$ time and answers queries in $\cO(\log n / \log \log n)$ time.
  \label{theorem:final}
\end{theorem}

\begin{table}[t]
\centering
\renewcommand{\arraystretch}{1.2}
\begin{tabular}{| c | c | c |}
  \hline
 {\bf Size} & {\bf Query time} & {\bf Paper}\\ \hline
  $\cOtilde(n)$ & $\Omega(\log n/\log \log n)$ & \cite{charalampopoulos_et_al} \\ \hline
  $\cO(n)$ & $\cO ( \log^2 n / \log\log n )$ & \cite{DBLP:conf/cpm/AbedinH0T18} \\ \hline
  $\cO(n\log n)$ & $\cO(\log n \log\log n)$ & \cite{DBLP:conf/cpm/AbedinH0T18} \\ \hline
  $\cO(n\log^2 n)$ & $\cO(\log n)$ & sketched in \cite{DBLP:conf/cccg/GagieGN13}, see \cref{app:fc} \\ \hline
  $\cOtilde(n)$ & $\cO(\log n/ \log\log n)$ & \emph{this work}\\ \hline
  $\cO(n^{1+\epsilon})$ & $\cO(1)$ & \cite{charalampopoulos_et_al} \\ \hline
\end{tabular}
\caption{Size vs.~query-time tradeoffs for the HIA problem; the size is measured in machine words.}
\label{table:results}
\vspace{-.25cm}
\end{table}

\subparagraph{Applications of HIA.}
Before discussing some concrete applications of the HIA problem in string algorithms and the consequences of our results for them, we
give a high-level description of how the HIA problem comes up in variants of computing an LCS.

Consider a string $S$ and a chosen subset $A$ of its positions, that we call \emph{anchors}.
Further, consider the following two tries: a trie $\mathcal{T}^{\leftarrow}$ for the strings in $\{S[1\dd k-1]^R: k \in A\}$, where~$U^R$ denotes the reversal of $U$,
and a trie $\mathcal{T}^{\rightarrow}$ for the strings in $\{S[k\dd |S|]: k \in A\}$.
In other words, for every anchor $k \in A$, we have a path in the
first trie for every prefix of $S[1 \dd k-1]^R$ and a path in the second trie for every prefix of $S[k \dd |S|]$.
We label each leaf of the two tries with the anchor it corresponds to.
Now, observe that a substring $S[i \dd j]$ that crosses an anchor $k$, i.e., $i < k \le j$,
corresponds to an induced pair of nodes in the tries. Indeed, there is a path
representing $S[i \dd k-1]^R$ in the first trie and a path representing
$S[k \dd j]$ in the second trie. An illustration of this idea is provided in
Figure \ref{fig:anchoring}.
The set of anchors and the HIA queries performed in an application of this technique depends on the
specific problem it is used for.
For some of the usages, one may consider using a compressed form of tries \cite{patricia}.

\begin{figure}[t]
\begin{center}
\scalebox{1}{\usetikzlibrary{arrows}
\usetikzlibrary{shapes}
\begin{tikzpicture}
    
\end{tikzpicture}
\begin{tikzpicture}[scale=0.35]
    \draw (0,0) -- (-8.5,5) -- (-8.5,-5) -- (0,0);
    \draw (0,0) -- (8.5,5) -- (8.5,-5) -- (0,0);
    \node at (-2,2) {$\mathcal{T}^{\leftarrow}$};
    \node at (2,2) {$\mathcal{T}^{\rightarrow}$};
    \node[fill,circle,inner sep=0.01cm,minimum size=0.1cm,label={[label distance=0cm]270:{$u$}}] (U) at (-5,1) {};
    \node[fill,circle,inner sep=0.01cm,minimum size=0.1cm,label={[label distance=0cm]270:{$v$}}] (V) at (5,-1) {};
    \draw[thick,magenta] (0,0) .. controls +(-1.5,-0.5) and +(1.5,0.5) .. (U);
    \draw[thick,green] (0,0) .. controls +(1.5,0.5) and +(-1.5,-0.5) .. (V);
    \node[fill,circle,inner sep=0.01cm,minimum size=0.1cm,label={[label distance=-0.1cm]30:{\color{red}$k$}}] (l1) at (-8.5,1.5) {};
    \node[fill,circle,inner sep=0.01cm,minimum size=0.1cm,label={[label distance=-0.1cm]150:{\color{red}$k$}}] (r1) at (8.5,-0.5) {};
    \draw[dotted] (U) .. controls +(-1,0.5) and +(1,-0.5) .. (l1);
    \draw[dotted] (V) .. controls +(1,0.5) and +(-1,-0.5) .. (r1);
    \draw (U) -- (-8.5,3.5) -- (-8.5,-1.5) -- (U);
    \draw (V) -- (8.5,1.5) -- (8.5,-3.5) -- (V);
\end{tikzpicture}}
\caption{An illustration of the anchoring technique for LCS computation,
with the constructed tries $\mathcal{T}^{\leftarrow}$ and $\mathcal{T}^{\rightarrow}$ drawn so that their roots are attached (in the middle).
Any substring anchored at $k$, can be obtained by reading in a left-to-right manner the edge-labels from some node $u \in \mathcal{T}^{\leftarrow}$ to some node $v \in \mathcal{T}^{\rightarrow}$, that both have a leaf-descendant labelled with $k$.}
\label{fig:anchoring}
\end{center}
\vspace{-.5cm}
\end{figure}
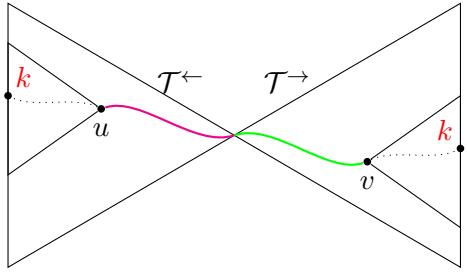

As a first application, consider the maintenance of an LCS of a static string $T$ and a dynamic string~$S$.
By plugging our HIA data structure into the approach of \cite{charalampopoulos_et_al}, we obtain the following result, improving the state-of-the-art by $\polyloglog n$ factors,
and matching the lower bound for the update-time when nearly-linear space is available \cite[Theorem 1]{charalampopoulos_et_al}.

\begin{corollary}\label{cor:onesided}
We can maintain an LCS of a dynamic string $S$ and a static string $T$, each of length at most $n$, in $\cO(\log n/\log\log n)$ time per substitution operation using $\cOtilde(n)$ space, after an $\cOtilde(n)$-time preprocessing.
\end{corollary}

Further, the authors of \cite{DBLP:conf/cccg/GagieGN13} (implicitly) reduced to the HIA problem,
the problem of preprocessing a text given in LZ77 compressed form
so that one can compute its LCS with uncompressed patterns given online.
Our HIA data structure yields the following result.

\begin{corollary}\label{cor:LZ_LCS}
Let $S$ be a string of length $N$ whose
LZ77 parse consists of~$n$ phrases. We can store~$S$ in
$\cO(n \log N + n \polylog n)$ space such that, given a pattern $P$ of length $m$,
we can compute the LCS of $S$ and $P$ in $\cO(m \log n/\log \log n)$ time. For each pattern $P$,
the returned result may be (consistently) incorrect with probability inverse polynomial in $n$.\footnote{Randomization is only used in the construction; all queries for the same pattern give identical results.}
\end{corollary}
\noindent
Other applications of the HIA problem in string algorithms can be found in~\cite{DBLP:conf/cpm/AbedinH0T18}.

\subparagraph{Tree Decompositions.}

One of the obvious divide-and-conquer techniques for efficiently solving algorithmic problems 
on trees is that of decomposing the tree(s) into smaller pieces and treating each of them separately.
The most important attributes of a tree decomposition are usually its depth, i.e.,
the maximum number of pieces that one path can intersect, and the structure of each
individual piece (e.g., pieces being paths may offer an advantage).
We next describe some tree decompositions for a tree $T$ with $n$ nodes. For a node $v$, denote by $s(v)$ the number of nodes in $v$'s subtree.

Arguably, the most well-known tree decomposition is the \emph{heavy-path decomposition}
\cite{HarelT84}.
Abstractly, this decomposition is a partition of the edges into \emph{light} and \emph{heavy},
such that:
\begin{itemize}
\item all connected components after deleting the light edges are paths, called \emph{heavy paths};
\item each root-to-leaf path consists of $\cO(\log n)$ prefixes of heavy paths and $\cO(\log n)$ light edges,
i.e., the depth of the decomposition is $\cO(\log n)$.
\end{itemize}
A heavy-path decomposition can be realized in several ways; two of which are as follows:
\begin{itemize}
\item HP1: Each non-leaf node $u$ of the tree chooses a child $v$ with maximum $s(v)$
and the edge from $u$ to $v$ is designated as heavy. The remaining edges outgoing from $u$ are light.\label{hp:1}
\item HP2: An edge $(u,v)$ is designated as heavy if and only if $\lfloor \log s(u) \rfloor = \lfloor \log s(v) \rfloor$.\footnote{In some works this has been called a centroid decomposition \cite{DBLP:journals/siamcomp/ColeH05}. It should not be confused with the hierarchical decomposition of the tree obtained by recursively deleting a centroid node, that is, a node whose removal splits the tree into three roughly equal components \cite{DBLP:conf/icalp/BrodalFPO01,DBLP:conf/spire/GiustinaPV19}.}\label{hp:2}
\end{itemize}
Intuitively, using a heavy-path decomposition, one may often lift
an algorithm that only works for paths and/or balanced trees
to work for arbitrary trees---usually with some overhead.

All previous works on the HIA problem used heavy-path decompositions, which,
as discussed, are of depth $\Omega(\log n)$. This adversely affects their query times as
one may have to traverse the decomposition along a root-to-leaf path at query time.
Thus, in order to achieve sublogarithmic query time, we considered tree decompositions
of smaller depths. There are a couple of generalizations of heavy-path decompositions that
have the sought depth, i.e., $\cO(\log n/ \log\log n)$.
We next discuss two such decompositions that are also based on partitioning the edges into light and heavy.
The caveat is that, for each of them, the connected components after the removal of the light edges are trees,
which we call \emph{heavy trees}, instead of paths and hence some extra work is required.\footnote{Heavy trees are sometimes called \emph{micro trees}, while the tree obtained from $T$ by contracting each micro tree
is called a \emph{macro tree}. We avoid this notation to not confuse with the so-called \emph{micro-macro decomposition} \cite{DBLP:journals/ipl/AlstrupSS97}, which, for a positive integer $k \leq n$, is a partition of the vertices of $T$ into $\cO(n/k)$ sets, such that each set $S$ is of size $\cO(k)$ it induces a subtree of $T$ and has at most two vertices that have neighbours that are not in $S$.}

The heavy $\alpha$-tree decomposition, introduced by Bille et al.~\cite{DBLP:journals/mst/BilleGVV14},
is of depth $\cO(\log_{\alpha} n)$ and is defined analogously to HP1:
each non-leaf node $u$ chooses its (at most) $\alpha$ heaviest (with respect to subtree-sizes) children;
the edge from $u$ to each of these children is designated as heavy, while all remaining edges outgoing from $u$ are designated as light.
By setting $\alpha=\lfloor \log n \rfloor$ one gets the sought depth.

An alternative is the so-called ART decomposition due to Alstrup et al.~\cite{alstrup}, which, for an input
integer parameter $b$, has depth $\cO(\log_b n)$.
For ease of presentation, we consider~$b$ to be equal to $\lfloor \log n \rfloor$ so that the depth of the decomposition is $\cO(\log n / \log\log n)$.
A partition of the edges yields such an ART decomposition if and only if
each heavy tree contains $\cO(\log n)$ nodes that have more than one child (in the heavy tree).
Alstrup et al.~\cite{alstrup} showed how to compute an ART decomposition by computing a set of leafmost light edges (in the spirit of HP2 with the base of the logarithm changed from 2 to $\lfloor \log n \rfloor$), removing them along with their descendants from the tree, and recursing.
Here, for convenience, we compute an ART decomposition similar to the HP2-realization of a heavy-path decomposition:
an edge $(u,v)$ is heavy if and only if both $s(u)$ and $s(v)$ are in $(n/{\lfloor \log n \rfloor}^{k+1},n/{\lfloor \log n \rfloor}^k]$.
For each heavy tree, we call \emph{branches} its maximal down-the-tree paths in which all nodes except the deepest one
have exactly one child (in the heavy tree); each heavy tree has $\cO(\log n)$ branches.

\subparagraph{Our techniques.} In order to answer an HIA query for nodes $u$ and $v$, we consider $\cO(\log n/ \log \log n)$ pairs of heavy trees that consist of a heavy tree in the root-to-$u$ path
and a heavy tree in the root-to-$v$ path.
For each such pair, we compute an induced pair $(x,y)$ of nodes in these heavy trees
that are ancestors of $u$ and $v$, respectively, and have maximum total weight.
Similarly to previous work, we observe that not every pair of heavy trees needs to be considered.
Instead, it suffices to consider a number of pairs of heavy trees linear to the depth of the tree decompositions
by a procedure analogous to the natural algorithm for checking whether there are two elements of a sorted list
that sum to a target $t$: start with two pointers, one at the beginning of the list and one at the end and move each of them in only one direction (either to the right or left).
For each pair of heavy trees, we construct data structures that can efficiently handle each of the cases of how the locations
of the lowest ancestors of $u$ and $v$ in the heavy trees relate to the locations
of the lowest ancestors (in the heavy trees) of same-label leaves.
Each data structure considers similar cases as previous work, however now we are working with two trees instead of two paths,
and hence need to be more careful.
This way, we reduce an HIA query to $\cO(\log n/ \log\log n)$ predecessor queries.
By answering each of these predecessor queries independently, we obtain a data structure that answers HIA queries in $\cO(\log n)$ time.
Indeed, in our case, each predecessor query requires $\Omega(\log \log n)$ time to be answered independently.

However, crucially, we show how to design the data structures so that all predecessor queries need only two values:
the preorder number of $u$ or the preorder number of $v$.
This is achieved by reordering the trees so that heavy edges come last. Then, larger preorder numbers
correspond to a larger depth of the lowest ancestor on a branch.
This means that the combination of our techniques with fractional cascading
would yield a faster algorithm for answering all the predecessor queries; the first one for each queried value would take $\cO(\log\log n)$ time, while all subsequent ones would take $\cO(1)$ time each.
The final technical hurdle is that fractional cascading requires the so-called underlying catalog graph
to have polylogarithmic degree \cite{quingminjaja}.
The construction of such a graph is straightforward if $T_1$ and $T_2$ are of polylogarithmic degree:
roughly speaking, it suffices to consider the Cartesian product of two trees whose nodes represent branches and heavy trees of each of $T_1$ and $T_2$.
We overcome this difficulty in the general case by reducing the maximum degree of these trees prior to taking their Cartesian product while maintaining all of their desirable properties.

\section{Preliminaries}

We use $[n]$ to denote the set $\{1, 2, \ldots, n\}$.
Throughout the paper, we perform the same operations on $T_1$ and $T_2$ and define
objects in these trees, so we are going to use $T_\star$ to denote any of the
trees. Similarly, we are going to use $v_\star$ to denote a node $v_1 \in T_1$
or a node $v_2$ in~$T_2$ etc. as an abbreviation of writing that some property holds for $v_i$ for both $i\in\{1,2\}$.

\subparagraph*{Lowest Common Ancestor. }
LCA queries can be answered in constant time after an $\cO(n)$-time preprocessing \cite{HarelT84}.

\defDSproblem{\textsc{LowestCommonAncestor} (\textsc{LCA})}{A rooted tree $T$.}{What is the node of largest depth that is an ancestor of both $u$ and $v$?}

\subparagraph*{Predecessor query. }
For a static set $S$, a combination of $x$-fast tries~\cite{WILLARD198381} and deterministic dictionaries~\cite{ruzic} yields an $\cO(n)$-size data structure that can be built in $\cO(n)$ time and answers predecessor queries in $\cO(\log \log U)$ time deterministically (cf., \cite[Proposition 2]{DBLP:conf/cpm/0001G15}); this is optimal~\cite{DBLP:journals/siamcomp/Patrascu11}.

\defDSproblem{\Pred}{A set $S$ of $n$ integers from $[U]$.}{For a given integer $y$, what is the largest $x \in S$ such that $x \le y$?}

\subparagraph*{Range Minimum Query. }
RMQs can be answered in constant time after an $\cO(n)$-time preprocessing \cite{fischer-heun,DBLP:journals/algorithmica/DemaineLW14}.
By setting, for each $i$, $S[i]:=|U|-S[i]$, we get a structure for the symmetric \textsc{RangeMaximumQuery} problem.

\defDSproblem{\textsc{RangeMinimumQuery} (\textsc{RMQ})}{A sequence $S$ of $n$ integers from $[U]$.}{For given positions $i$ and $j$, with $1 \le i \le j \le n$, what is (the position of) the minimum among $S[i], S[i+1], \ldots, S[j]$?}

\subparagraph*{Deterministic Static Dictionary. }
A dictionary is a structure that stores a set of keys (often with associated values)
and allows answering membership queries (or getting the value of a given key).
There are multiple randomized solutions, but there 
is even a deterministic solution with
$\cO(n)$ space, $\cOtilde(n)$-time preprocessing, and constant-time queries~\cite{ruzic}.

\subparagraph*{Fractional cascading.}
Consider a directed graph $C$, called the \emph{catalog graph}, which has a sorted list (also called a catalog)
in each of its nodes.
Let the total size of the lists be $n$.
Now, suppose that we want to answer queries of the following type:
for a connected subgraph $G$ of~$C$ and a query value $v$,
find the predecessor of $v$ in each of the lists stored in the nodes of $G$.

A naive way of solving this problem would be to ignore any preprocessing and run
a separate binary search in the sorted list of each of the nodes of $G$, for a total of $\cO(|G| \log n)$ time.
Fractional cascading is a general optimization technique that allows the speed-up
of multiple binary searches for the same value over multiple related sorted
sequences of objects.

If the degree of each node of the
catalog graph is bounded by a constant, the original solution of Chazelle and Guibas
\cite{chazelle} answers a query in $\cO(\log n + |G|)$ time after a linear-time
preprocessing in the comparison model.
To be precise, it is sufficient for the catalog graph to have
\emph{locally bounded degree} (as per Definition 1 of \cite{chazelle}).
Unfortunately, in our case, this is not useful.
A subsequent work of Shi and JáJá \cite{quingminjaja} achieved the same complexities
for graphs of polylogarithmic maximum degree in the word RAM model of computation
(in the original description, the graph is a tree, however, there is no difficulty in extending this
to the general setting considered by Chazelle and Guibas~\cite{chazelle}).
Crucially, these data structures can also handle the case where the
nodes of $G$ are given one by one in an online manner;
the only requirement is that each node (other than the first) must be a neighbour of some previous one.
Note that the $\cO(\log n)$ term in the complexities comes from performing
a binary search in the first of the considered lists.
Then, the predecessor of the query value $v$ in each of the subsequently considered lists
is obtained by following a constant number of pointers, which can be retrieved in $\cO(1)$ time.
(During the preprocessing phase, the catalogs are augmented in an appropriate manner and said pointers are constructed.)
In the word RAM model of computation, the first query can be solved faster using other data
structures, e.g., in $\Oh(\log\log U)$ time with the structure discussed above for the \Pred problem.

\section{An $\cOtilde(n)$-size Data Structure with Optimal Query Time}

Let $b > 1$ be a parameter to be chosen later.
Consider a rooted, weighted, and labelled tree~$T$.
The weight of a node $u$ is denoted by $\weight(u)$.
For a node $v$, we denote by $s(v)$ the number of nodes in
$v$'s subtree, including $v$.
For an integer $k$,
a node~$v$ is on layer~$k$ if and only if $n/b^{k+1} < s(v) \le n/b^k$.
An edge that connects nodes of the same layer is called \emph{heavy} and the
other edges are \emph{light}.
Each maximal subtree that does not contain any light edges is called
a \emph{heavy tree}.
We stress that a heavy tree might be a singleton.

We decompose each heavy tree into \emph{branches}, that is, maximal down-the-tree paths of nodes, where every node apart from the deepest one has one child (in the heavy tree).
The last node is either a leaf of the heavy tree or has at least two children.
Note that there can be branches consisting of a single node.
We call a node \emph{implicit} if it is an internal (non-leaf) node of the heavy tree with one child; otherwise, we call it \emph{explicit}.
For a heavy tree, we obtain a \emph{compacted} version of it, called \emph{compacted heavy tree}, by eliminating all the implicit nodes through the contraction of either of their incident edges.
See an example in Figure \ref{fig:heavy_tree_expl}.

\begin{figure}[h]
\begin{center}
\scalebox{0.6}{\begin{tikzpicture}

\node[fill,circle,scale=0.5,label={[label distance=0.01cm]30:26}] (v-26) at (0,8) {};
\node[fill,circle,scale=0.5,label={[label distance=0.01cm]180:13}] (v-13) at (-1.5,7) {};
\node[fill,circle,scale=0.5,label={[label distance=0.01cm]0:12}] (v-12) at (1.5,7) {};
\node[fill,circle,scale=0.5,label={[label distance=0.01cm]180:6}] (v-6a) at (-2.5,6) {};
\node[fill,circle,scale=0.5,label={[label distance=0.01cm]180:6}] (v-6b) at (-0.5,6) {};
\node[fill,circle,scale=0.5,label={[label distance=0.01cm]0:10}] (v-10) at (1,6) {};
\node[fill,circle,scale=0.5,label={[label distance=0.01cm]0:1}] (v-1i) at (2.5,6) {};
\node[fill,circle,scale=0.5,label={[label distance=0.01cm]180:5}] (v-5) at (-2.5,5) {};
\node[fill,circle,scale=0.5,label={[label distance=0.01cm]180:4}] (v-4a) at (-2.5,4) {};
\node[fill,circle,scale=0.5,label={[label distance=0.01cm]180:3}] (v-3a) at (-2.5,3) {};
\node[fill,circle,scale=0.5,label={[label distance=0.01cm]180:2}] (v-2a) at (-2.5,2) {};
\node[fill,circle,scale=0.5,label={[label distance=0.01cm]180:1}] (v-1a) at (-2.5,1) {};
\node[fill,circle,scale=0.5,label={[label distance=0.01cm]180:3}] (v-3b) at (-1.5,5) {};
\node[fill,circle,scale=0.5,label={[label distance=0.01cm]270:1}] (v-1b) at (-2,4) {};
\node[fill,circle,scale=0.5,label={[label distance=0.01cm]270:1}] (v-1c) at (-1,4) {};
\node[fill,circle,scale=0.5,label={[label distance=0.01cm]270:1}] (v-1d) at (-0.75,5) {};
\node[fill,circle,scale=0.5,label={[label distance=0.01cm]270:1}] (v-1e) at (0,5) {};
\node[fill,circle,scale=0.5,label={[label distance=0.01cm]0:9}] (v-9) at (1,5) {};
\node[fill,circle,scale=0.5,label={[label distance=0.01cm]0:8}] (v-8) at (1,4) {};
\node[fill,circle,scale=0.5,label={[label distance=0.01cm]0:7}] (v-7) at (1,3) {};
\node[fill,circle,scale=0.5,label={[label distance=0.01cm]180:4}] (v-4b) at (0.25,2) {};
\node[fill,circle,scale=0.5,label={[label distance=0.01cm]0:2}] (v-2c) at (1.75,2) {};
\node[fill,circle,scale=0.5,label={[label distance=0.01cm]0:1}] (v-1h) at (1.75,1) {};
\node[fill,circle,scale=0.5,label={[label distance=0.01cm]180:2}] (v-2b) at (-0.25,1) {};
\node[fill,circle,scale=0.5,label={[label distance=0.01cm]180:1}] (v-1f) at (-0.25,0) {};
\node[fill,circle,scale=0.5,label={[label distance=0.01cm]180:1}] (v-1g) at (0.75,1) {};
\draw[ultra thick] (v-26) -- (v-13);
\draw (v-13) -- (v-6a);
\draw[ultra thick] (v-6a) -- (v-5);
\draw[ultra thick] (v-5) -- (v-4a);
\draw[ultra thick] (v-4a) -- (v-3a);
\draw (v-3a) -- (v-2a);
\draw[ultra thick] (v-2a) -- (v-1a);
\draw[ultra thick] (v-26) -- (v-12);
\draw[ultra thick] (v-12) -- (v-10);
\draw (v-12) -- (v-1i);
\draw[ultra thick] (v-10) -- (v-9);
\draw (v-9) -- (v-8);
\draw[ultra thick] (v-8) -- (v-7);
\draw[ultra thick] (v-7) -- (v-4b);
\draw (v-7) -- (v-2c);
\draw (v-4b) -- (v-2b);
\draw[ultra thick] (v-2b) -- (v-1f);
\draw (v-4b) -- (v-1g);
\draw[ultra thick] (v-2c) -- (v-1h);
\draw (v-13) -- (v-6b);
\draw[ultra thick] (v-6b) -- (v-3b);
\draw (v-6b) -- (v-1d);
\draw (v-6b) -- (v-1e);
\draw (v-3b) -- (v-1b);
\draw (v-3b) -- (v-1c);
\end{tikzpicture}

\hspace{2.5cm}

\raisebox{4cm}{
${\Longrightarrow}$
}

\hspace{2.5cm}

\raisebox{1.5cm}{
\begin{tikzpicture}
\node[fill,circle,scale=0.5,label={[label distance=0.01cm]30:26}] (v-26) at (0,8) {};
\node[fill,circle,scale=0.5,label={[label distance=0.01cm]180:13}] (v-13) at (-1.5,7) {};
\node[draw,circle,scale=0.5,label={[label distance=0.01cm]0:12}] (v-12) at (0.5,7.6667) {};
\node[draw,circle,scale=0.5,label={[label distance=0.01cm]180:6}] (v-6a) at (-2.5,6) {};
\node[draw,circle,scale=0.5,label={[label distance=0.01cm]180:6}] (v-6b) at (-0.5,6) {};
\node[draw,circle,scale=0.5,label={[label distance=0.01cm]0:10}] (v-10) at (1,7.3333) {};
\node[fill,circle,scale=0.5,label={[label distance=0.01cm]0:1}] (v-1i) at (2.5,7) {};
\node[draw,circle,scale=0.5,label={[label distance=0.01cm]180:5}] (v-5) at (-2.5,5.6667) {};
\node[draw,circle,scale=0.5,label={[label distance=0.01cm]180:4}] (v-4a) at (-2.5,5.3333) {};
\node[fill,circle,scale=0.5,label={[label distance=0.01cm]180:3}] (v-3a) at (-2.5,5) {};
\node[draw,circle,scale=0.5,label={[label distance=0.01cm]180:2}] (v-2a) at (-2.5,4) {};
\node[fill,circle,scale=0.5,label={[label distance=0.01cm]180:1}] (v-1a) at (-2.5,3) {};
\node[fill,circle,scale=0.5,label={[label distance=0.01cm]180:3}] (v-3b) at (-1.5,5) {};
\node[fill,circle,scale=0.5,label={[label distance=0.01cm]270:1}] (v-1b) at (-2,4) {};
\node[fill,circle,scale=0.5,label={[label distance=0.01cm]270:1}] (v-1c) at (-1,4) {};
\node[fill,circle,scale=0.5,label={[label distance=0.01cm]270:1}] (v-1d) at (-0.75,5) {};
\node[fill,circle,scale=0.5,label={[label distance=0.01cm]270:1}] (v-1e) at (0,5) {};
\node[fill,circle,scale=0.5,label={[label distance=0.01cm]0:9}] (v-9) at (1.5,7) {};
\node[draw,circle,scale=0.5,label={[label distance=0.01cm]0:8}] (v-8) at (1,6) {};
\node[draw,circle,scale=0.5,label={[label distance=0.01cm]0:7}] (v-7) at (1,5.5) {};
\node[fill,circle,scale=0.5,label={[label distance=0.01cm]0:4}] (v-4b) at (1,5) {};
\node[draw,circle,scale=0.5,label={[label distance=0.01cm]0:2}] (v-2c) at (1.75,4) {};
\node[fill,circle,scale=0.5,label={[label distance=0.01cm]0:1}] (v-1h) at (1.75,3) {};
\node[draw,circle,scale=0.5,label={[label distance=0.01cm]180:2}] (v-2b) at (0.25,4) {};
\node[fill,circle,scale=0.5,label={[label distance=0.01cm]180:1}] (v-1f) at (0.25,3) {};
\node[fill,circle,scale=0.5,label={[label distance=0.01cm]270:1}] (v-1g) at (1,4) {};
\draw (v-26) -- (v-13);
\draw (v-13) -- (v-6a);
\draw (v-6a) -- (v-5);
\draw (v-5) -- (v-4a);
\draw (v-4a) -- (v-3a);
\draw (v-3a) -- (v-2a);
\draw (v-2a) -- (v-1a);
\draw (v-26) -- (v-12);
\draw (v-12) -- (v-10);
\draw (v-12) .. controls +(1,1) and +(0,1) .. (v-1i);
\draw (v-10) -- (v-9);
\draw (v-9) -- (v-8);
\draw (v-8) -- (v-7);
\draw (v-7) -- (v-4b);
\draw (v-7) .. controls +(0.5,-0.5) and +(-1,0.5) .. (v-2c);
\draw (v-4b) -- (v-2b);
\draw (v-2b) -- (v-1f);
\draw (v-4b) -- (v-1g);
\draw (v-2c) -- (v-1h);
\draw (v-13) -- (v-6b);
\draw (v-6b) -- (v-3b);
\draw (v-6b) -- (v-1d);
\draw (v-6b) -- (v-1e);
\draw (v-3b) -- (v-1b);
\draw (v-3b) -- (v-1c);

\draw[dotted] (-0.1,7.5) ellipse (2.35 and 1.2);
\draw[dotted] (-2.6,5.5) ellipse (0.5 and 0.8);
\draw[dotted] (-2.7,3.5) ellipse (0.5 and 0.8);
\draw[dotted] (1.1,5.5) ellipse (0.5 and 0.8);
\draw[dotted] (0.15,3.5) ellipse (0.5 and 0.8);
\draw[dotted] (1.95,3.5) ellipse (0.5 and 0.8);
\draw[dotted] (0.15,3.5) ellipse (0.5 and 0.8);
\draw[rotate around={-45:(-1.15,5.5)},dotted] (-1.15,5.5) ellipse (0.45 and 1.1);
\draw[dotted] (v-1b) ellipse (0.25 and 0.55);
\draw[dotted] (v-1c) ellipse (0.25 and 0.55);
\draw[dotted] (-0.75,4.8) ellipse (0.25 and 0.4);
\draw[dotted] (v-1e) ellipse (0.25 and 0.55);
\draw[dotted] (v-1g) ellipse (0.25 and 0.55);
\draw[dotted] (v-1i) ellipse (0.25 and 0.55);
\end{tikzpicture}
}}
\caption{An example tree with $n = 26$ is shown.
Each node is labelled with its subtree size, while $b=3$.
On the left, the heavy edges are thick, while the light edges are thin.
On the right, the tree is decomposed into heavy trees (marked with
ellipses).
The compaction of heavy trees is illustrated as follows: empty circles denote implicit nodes, while full circles denote explicit nodes.
}
\label{fig:heavy_tree_expl}
\end{center}
\vspace{-0.2cm}
\end{figure}
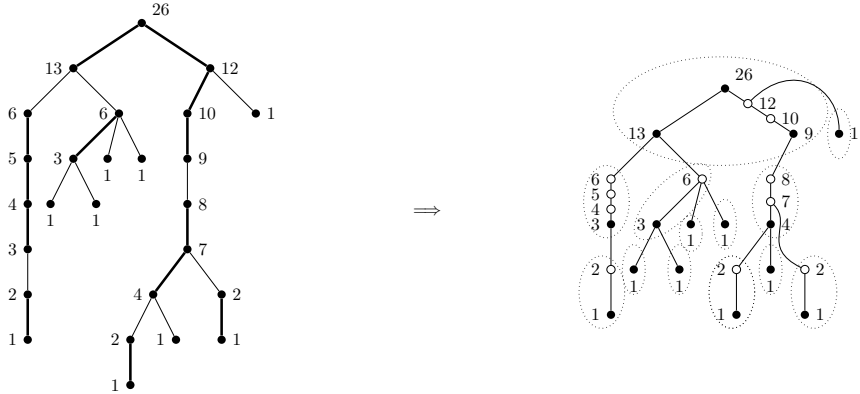

Observe that in a compacted tree, every non-leaf node has at least two
children.
Hence, there are fewer internal nodes than there are leaves of the tree.
As every leaf in a heavy tree has a sufficiently big subtree underneath it, we obtain a bound on the total number of nodes inside a single compacted heavy tree;
this shows that we obtain an ART decomposition \cite{alstrup}.

\begin{lemma}
There are at most $\cO(\log n / \log b)$ layers in a tree $T$ on $n$ leaves and each heavy tree has $\cO(b)$ branches.
\label{lemma:layers_count}
\end{lemma}
\begin{proof}
  Consider a heavy tree $H$ of $T$ and its compacted form $H_C$.
  Let $r$ be the root of $H$.
  For every leaf $\ell$ of $H$, we have that $s(\ell)> s(r)/b$ as $\ell$ and $r$ are nodes of the same heavy tree, and hence they are in the same layer of $T$.
  As the subtrees of $T$ rooted at the leaves of~$H$ are disjoint and their total size is at most $s(r)$, $H$ contains at most $b$ leaves.
  Further, as there are no internal nodes with one child in $H_C$, there are
  at most $b-1$ non-leaf nodes in $H_C$, and hence $H_C$ has $\cO(b)$ nodes.
  All branches in $H$ are disjoint, so there are $\cO(b)$ of them.

  Consider the $k$-th layer of the tree $T$ and a node $u$ in this layer.
  As the subtree of $u$ has at least one node, we have $1\leq s(u) \leq n/b^k$, so $k\leq \log_b n = \log n / \log b$.
\end{proof}

For an HIA query $(v_1, v_2)$, the paths from $v_1$ and $v_2$ to the roots of their
respective trees are called \emph{query paths}.
The result of an HIA query is a pair of nodes $(u_1, u_2)$, that are on the query paths
from $v_1$ and $v_2$, respectively.
To answer a query $(v_1, v_2)$, we identify the sequences of heavy trees $B_1$ and $B_2$ that contain nodes on the query paths from $v_1$ and $v_2$, respectively, and perform \emph{restricted HIA queries} for some pairs of those heavy trees.
More precisely, in each step of the algorithm, having chosen two heavy trees
$B_1[i]$ and $B_2[j]$, we try to find the pair of induced ancestors $(u_1, u_2) \in B_1[i] \times B_2[j]$
of $(v_1, v_2)$
with the maximum combined weight or determine that there is no such pair.
A pseudocode for this procedure is given as \cref{alg:hia}.

\begin{algorithm}
    \DontPrintSemicolon
    \SetKwProg{Fn}{function}{}{}
    \SetKwFunction{hia}{hia}
    \SetKwFunction{length}{length}
    \SetKwFunction{weight}{weight}
    \SetKwFunction{restrictedhia}{restricted-hia}
    \SetKwData{null}{null}
    \SetKwData{false}{false}

    \Fn{\hia{$T_1$, $T_2$, $v_1$, $v_2$}}{
      $(r_1, r_2) \gets \null$ \;
      $(B_1, B_2) \gets \textrm{sequences of heavy trees on paths from $v_1$ and $v_2$ down-the-tree}$ \;
      $i \gets 0$ \;
      $j \gets |B_2| - 1$ \;
      \Repeat{\false}{
        $x_1 \gets \textrm{lowest ancestor of $v_1$ in $B_1[i]$}$ \;
        $x_2 \gets \textrm{lowest ancestor of $v_2$ in $B_2[j]$}$ \;
        $(u_1, u_2) \gets \restrictedhia{$v_1$, $v_2$, $B_1[i]$, $B_2[j]$, $x_1$, $x_2$}$\ \;
        \If{$(r_1, r_2) = \null \mathbf{~or~} \weight{$u_1$} + \weight{$u_2$} > \weight{$r_1$} + \weight{$r_2$}$}{
          $(r_1, r_2) \gets (u_1, u_2)$ \;
        }
        \If{$(i = |B_1| - 1) \mathbf{~and~} (j = 0)$}{\Return{$(r_1, r_2)$}}
        \;

        \If{$i+1 < |B_1| \mathbf{~and~} \textrm{roots of $B_1[i+1]$ and $B_2[j]$ are induced}$}{
          $i \gets i + 1$ \;
        }
        \Else{
          $j \gets j - 1$ \;
        }
      }
    }
    \caption{Algorithm for answering HIA queries.}
    \label{alg:hia}
\end{algorithm}

To make the description more modular, we provide $x_1,x_2$ to the restricted HIA query, where $x_i$ is the lowest ancestor of $v_i$ in the considered heavy tree for $i\in\{1,2\}$.
Note that $x_1$ is either the parent of the root of $B_1[i+1]$ if $i+1<|B_1|$ or $x_1$ equals to $v_1$ otherwise, and similarly for~$x_2$.
We note that in Algorithm \ref{alg:hia}, we ask restricted HIA queries about the same
pair $(v_1, v_2)$ for various pairs $(B_1[i], B_2[j])$ of heavy trees and it may be that
$v_1 \not \in B_1[i]$ and/or $v_2 \not \in B_2[j]$, whereas we also provide nodes $x_\star$, in which $v_{\star}$ connects to the respective heavy subtree in $B_{\star}$

\begin{lemma}
  Algorithm \ref{alg:hia} performs $O(\log n / \log b)$ restricted HIA queries
  to find the heaviest induced ancestors of nodes $u$ and $v$.
  \label{lemma:tree_height}
\end{lemma}
\begin{proof}
  The paths from $v_\star$ to the roots pass heavy trees with monotonically decreasing indices of layers, so the sequences $B_1$ and $B_2$ contain at most $\cO(\log n / \log b)$ elements.
  After every restricted HIA query, we either increase $i$ or decrease $j$ by $1$, so the total number of restricted HIA queries that we perform is at most $|B_1|+|B_2|=\cO(\log n / \log b)$.

  The correctness follows from the monotonicity of being induced.
  For an induced pair of nodes, any pair of their (weak) ancestors is also induced.
  Conversely, if a pair is not induced, any pair of their (weak) descendants
  is also not induced.

  This implies that if the roots of $B_1[i]$ and $B_2[j]$ are induced then there is no need to query for any pair
  $(B_1[i'], B_2[j'])$ with $(i' < i) \wedge (j' < j)$,
  as such a query would return a pair of (strict) ancestors of the pair
  returned by the restricted HIA query for $(B_1[i], B_2[j])$.
  We call a pair of trees $(B_1[i'], B_2[j'])$ \emph{dominated}, if there exist $i>i'$ and $j>j'$
  such that the roots of $B_1[i]$ and $B_2[j]$ are induced.
  We show that the algorithm performs restricted HIA queries at Line 9 exactly for those pairs of heavy trees
  that (i) are not dominated and (ii) for which
  the result of the restricted HIA query is not \textsf{null}.

  The algorithm maintains the invariant that each of the restricted HIA queries is called for the pairs of trees
  for which their roots are induced
  and that the pair $(B_1[i], B_2[j])$ is not dominated.
  This is true for the first iteration, where the pair $(B_1[0], B_2[|B_2|-1])$
  is considered, because $B_2[|B_2|-1]$ has a leaf and the root of $B_1[0]$ is the root of $T_1$.
  Now we show that the invariant is maintained later.
  We start with the assumption that the result of $\restrictedhia(v_1, v_2, B_1[i], B_2[j],x_1,x_2)$
  is not $\textsf{null}$ and this pair is non-dominated.
  Now, we have to distinguish between two cases.
  \enumcases{
  \item We next consider pair $(B_1[i+1], B_2[j])$.
  It means that the check in Line 15 confirmed that the roots of the the trees are induced, so the the result of calling $\restrictedhia(v_1, v_2, B_1[i+1], B_2[j],x_1,x_2)$  is not $\textsf{null}$. Further, this pair of heavy paths is not dominated since $(B_1[i], B_2[j])$ is not dominated.
  \item We next consider pair $(B_1[i], B_2[j-1])$. This can only happen if $i=|B_1|-1$ or the roots of $B_1[i+1]$ and $B_2[j]$ are not induced; in either of these cases, $(B_1[i], B_2[j-1])$ is not dominated. Further, as the answer to the HIA query for pair $(B_1[i],B_2[j])$ is not $\textsf{null}$, the answer for $(B_1[i], B_2[j-1])$ cannot be $\textsf{null}$ either.
  }
  Thus, the invariant is maintained in both cases.
  
  Clearly, the heaviest induced pair of ancestors of $v_1,v_2$ belongs to a pair of heavy trees that satisfy conditions (i) and (ii).
  We claim that we process all such pairs.
  Observe that for a fixed $j$ there are two indices $0\leq i_1<i_2\leq |B_1|$ such that the pair $B_1[i],B_2[j]$ is dominated for $0\leq i <i_1$, non-dominated for $i_1\leq i <i_2$, and corresponds to a $\textsf{null}$ answer for $i_2\leq i < |B_1|$.
  By the invariant, just after any decrease of $j$ in Line 18 it holds that $i\geq i_1$, as $B_1[i],B_2[j]$ is non-dominated.
  Actually $i=i_1$, because the pair $B_1[i-1],B_2[j]$ is dominated as in the previous step we considered the pair $B_1[i],B_2[j+1]$ for which the answer was not $\textsf{null}$.
  As in the next steps we process all $i$ up to (but excluding) $i_2$, the claim follows.
\end{proof}

\subsection{Restricted HIA Queries}
In this subsection, we present a data structure that efficiently answers restricted HIA queries.
\begin{theorem}\label{thm:main_data_structure}
 For every two trees $T_1,T_2$ on $n$ leaves and an integer parameter $b \in [n]$, there exists an $\cO(nb^2 \log^2 n / \log^2 b)$-size data structure that can be computed in $\Ohtilde(nb^2/\log^2 b)$ time and answers
 (i) queries about whether the roots of two given heavy trees are induced in constant time,
 (ii) any restricted HIA query $\restrictedhia(v_1, v_2, H_1, H_2,x_1,x_2)$ in constant time plus the time required to answer a predecessor query about $\pre(v_1)$ and one about $\pre(v_2)$; these predecessor queries are performed on two out of $\cO(b^2)$ (preprocessed) lists stored for the pair of heavy trees $(H_1,H_2)$.
\end{theorem}

We divide the proof into three parts: we first describe the preprocessing phase, then discuss the properties of the created data structure, and, finally, present the query procedure.
\subparagraph{Preprocessing.}
First, we compute the partition of the edges of each of $T_1,T_2$ into heavy and light, and the implied heavy trees.
For each node, we store its assignment to the heavy tree and to the branch to which it belongs.
Further, for each $T_\star$, we build a linear-size data structure for answering LCA queries in $\cO(1)$ time \cite{HarelT84}.
For each node in $T_\star$, we fix the order of its children such that the children that are in the same heavy tree are last.
(The order of children that are connected to the parent with the same type of edge, heavy or light, is arbitrary.)
Next, we compute preorder traversals of $T_\star$, for each node $u$, we denote by $\pre(u)$ the preorder number of $u$ and by $T_\star[p]$ the node of $T_\star$ whose preorder number is $p$; we have $T_\star[\pre(u)]=u$.
Additionally, for each label, we identify the leaves of $T_1$ and $T_2$ with that label (recall the labels in a single tree are pairwise distinct).

Next, for each pair $(\ell_1,\ell_2)$ of leaves with the same label, we iterate over all pairs $(B_1^{\ell_1}[i_1],B_2^{\ell_2}[i_2])$ of heavy trees on their query paths and insert a point to the data structure for each pair of branches in $B_1^{\ell_1}[i_1]\times B_2^{\ell_2}[i_2]$. This procedure is formalized as Algorithm~\ref{alg:preprocessing}.
\begin{algorithm}
    \DontPrintSemicolon
    \SetKwProg{Proc}{procedure}{}{}
    \SetKwFunction{addlabel}{add\_label}

    \Proc{\addlabel{$T_1$, $T_2$, $\ell_1$, $\ell_2$}}{
      $(B_1^{\ell_1}, B_2^{\ell_2}) \gets \textrm{sequence of heavy trees on query paths from $\ell_1$ and $\ell_2$}$ \;
      \For{$i_1 \gets 0, 1, \dots, |B_1^{\ell_1}| - 1$}{
        \For{$i_2 \gets 0, 1, \dots, |B_2^{\ell_2}| - 1$}{
          \For{$e_1 \gets \textrm{branch of $B_1^{\ell_1}[i_1]$}$}{
            \For{$e_2 \gets \textrm{branch of $B_2^{\ell_2}[i_2]$}$}{
              $w_1 \gets \LCA(\ell_1, \textrm{lowest node on $e_1$})$\;
              $w_2 \gets \LCA(\ell_2, \textrm{lowest node on $e_2$})$\;
              \textrm{insert point $(\pre(w_1), \pre(w_2))$ to $D^{B_1^{\ell_1}[i_1],B_2^{\ell_2}[i_2]}[e_1, e_2]$} \;
            }
          }
        }
      }
    }
    \caption{Preprocessing for a pair of leaves $(\ell_1, \ell_2)$ with the same label.}
    \label{alg:preprocessing}
\end{algorithm}

We call pairs $(B_1,B_2)$ of heavy trees that are processed by \cref{alg:preprocessing} \emph{relevant}.
We first run this algorithm once just to record all relevant pairs of heavy trees, without inserting any points to any structures.
We then sort the relevant pairs, remove duplicates, and construct a deterministic dictionary over them~\cite{ruzic}.
This allows us to check in constant time if the roots of two trees are induced because this is equivalent to checking if the pair of trees is relevant.
For each relevant pair of heavy trees, we initialize an array $D^{B_1,B_2}$ indexed by pairs of branches $(e_1, e_2)$, where $e_1$ is a branch in $B_1$ and $e_2$ is a branch in $B_2$.
In each entry of the array, we create (store a pointer to) a data structure for the corresponding pair of branches.
We then re-run \cref{alg:preprocessing}, inserting the points to the structures as needed, with the help of the deterministic dictionary built for relevant pairs of heavy trees.
As each branch~$e_\star$ belongs to a unique heavy tree, we often drop the superscript and write $D[e_1,e_2]$ instead of $D^{B_1,B_2}[e_1,e_2]$.
We call a pair $(e_1,e_2)$ of branches \emph{relevant} if and only if pair of their assigned heavy trees is relevant.
By Lemma~\ref{lemma:layers_count}, every query path is decomposed into $\cO(\log n / \log b)$ parts on different layers and each heavy tree has $\cO(b)$ branches.
Hence, for every pair $(\ell_1,\ell_2)$ of leaves with the same label, we insert a point to $\cO(b^2 \log^2 n / \log^2 b)$ structures.

Finally, for each relevant pair of branches $(e_1,e_2)$, we perform the following postprocessing of structure $D[e_1,e_2]$:

\begin{itemize}
 \item Remove all points $(x,y)$ for which there exists another point $(x',y')$ such that $x\leq x',y\leq y'$ and $(x,y)\ne(x',y')$.
 This can be done in $\Ohtilde(|D[e_1,e_2]|)$ time by sorting the points and processing them in the left-to-right order.
 \item Let $D_x[e_1,e_2]$ and $D_y[e_1,e_2]$ be the sets of $x$- and $y$-coordinates of the remaining points, respectively.
 We build a data structure for the \Pred problem for each of $D_x[e_1,e_2]$ and $D_y[e_1,e_2]$ separately.
 \item We build a data structure for the \textsc{RangeMaximumQuery} problem for the points remaining in $D[e_1,e_2]$ sorted by $x$-coordinate, where the weight of a point $(x,y)$ is $\weight(T_1[x])+\weight(T_2[y])$.
\end{itemize}
\noindent
We call the above stage the postprocessing of $D[e_1,e_2]$.

To summarize the whole preprocessing stage for trees $T_\star$, for each of the $n$ labels we add $\cO(b^2 \log^2 n / \log^2 b)$ points to structures $D[\cdot,\cdot]$,
for a total number of $\cO(nb^2 \log^2 n / \log^2 b)$ points.
The postprocessing of all the structures $D[\cdot,\cdot]$ takes nearly linear time in their size and hence the total running time is $\cOtilde(nb^2/\log^2 b)$.
The structures for the \Pred and \textsc{RangeMaximumQuery} problems have size linear in the number of elements they are built over and hence the total space is $\cO(nb^2 \log^2 n / \log^2 b)$.

\subparagraph{Properties of structures $\mathbf{D[e_1,e_2]}$.}
In this paragraph, we show some properties of the structures $D[e_1,e_2]$ that are useful for answering restricted HIA queries efficiently.

\begin{property}\label{prop:adding_points_on_path}
 For every pair $(w_1,w_2)$ added to $D^{B_1^{\ell_1}[i_1],B_2^{\ell_2}[i_2]}[e_1, e_2]$, $w_\star$ is either on $e_\star$ or on the path from the highest node of $e_\star$ to the root of $B_\star^{\ell_\star}[i_\star]$.
\end{property}
\begin{proof}
Recall that $w_\star$ is the lowest common ancestor of $\ell_\star$ and the lowest node $q$ on $e_\star$.
Observe that as $B_\star^{\ell_\star}[i_\star]$ is on the query path from $\ell_\star$, the root $r$ of $B_\star^{\ell_\star}[i_\star]$ is an ancestor of both $\ell_\star$ and $q$.
Hence, $w_\star$ lies on the $r$-to-$q$ path, which directly yields the statement.
\end{proof}
\noindent
Note that after the first step of postprocessing, $D[e_1,e_2]$ satisfies the following property:

\begin{property}\label{prop:monotonicity}
  After the postprocessing, for every pair $(e_1,e_2)$ of branches, after sorting the points of $D[e_1,e_2]$ increasingly by $x$-coordinate, the sequence of points is also sorted decreasingly by $y$-coordinate. 
\end{property}
\noindent
Informally, we can now consider a one-dimensional problem, with points forming a sequence that can be efficiently navigated both in $x$- and $y$-coordinates via predecessor queries.

We next show how the computed data structures $D[e_1,e_2]$ enable us to answer restricted HIA queries efficiently.

\subparagraph{Answering a restricted HIA query.}
\label{para:answeringrhia}
We are now ready to present how to answer a restricted HIA query for a pair $(v_1,v_2)$ of nodes and heavy trees $B_1$ and $B_2$ on the query paths from $v_1$ and~$v_2$.
Let $(r_1,r_2)=\restrictedhia(v_1,v_2,B_1,B_2,x_1,x_2)$ be a pair of ancestors of $v_1$ and $v_2$ within the trees $B_1$ and $B_2$ that are induced and have the maximum total weight.
Recall that $x_\star$ is the lowest weak ancestor of $v_\star$ that is in $B_\star$ and let $e_\star$ be the branch containing~$x_\star$.
Further, let $\ell$ be the label inducing $(r_1,r_2)$ and let leaves $\ell_\star$ share this label.

First, we show that we can find an induced pair of ancestors of $v_1$ and $v_2$ with the maximum combined weight using the structure $D[e_1,e_2]$ before postprocessing.
Then, we show that after the postprocessing stage, we can still retrieve the correct answer but more efficiently, by performing predecessor queries for $\pre(x_1)$ and $\pre(x_2)$.
Finally, we show that we can call predecessor queries for $\pre(v_1)$ and $\pre(v_2)$ instead of $\pre(x_1)$ and $\pre(x_2)$.
The last step is not important for the correctness or efficiency of a single restricted HIA query but improves the complexity of Algorithm~\ref{alg:hia}.
Indeed, as all predecessor queries are for one of $\pre(v_1)$ or $\pre(v_2)$, we can use fractional cascading.
We explain this final component in detail in Section~\ref{subsec:fractional_cascading}.

Recall that in Algorithm~\ref{alg:preprocessing}, we insert point $(\pre(w_1),\pre(w_2))$ to $D[e_1,e_2]$, where $w_\star=\LCA(\ell_\star, \textrm{lowest node on }e_\star)$.
In the proof of Property~\ref{prop:adding_points_on_path}, we mention that $w_\star$ always belongs to $B_\star$ as the root of $B_\star$ is an ancestor of both $\ell_\star$ and the lowest node on $e_\star$.
There are two possible relative locations of $w_\star$ and $x_\star$ within a heavy tree:
\begin{itemize}
 \item $\ell$ is \emph{below} $x_\star$ when $w_\star$ is a (not necessarily proper) descendant of $x_\star$;
 \item $\ell$ is \emph{attached above} $x_\star$ when $w_\star$ is a proper ancestor of $x_\star$.
\end{itemize}
There are four cases for the relative locations of $\ell$ with respect to $x_1$ and $x_2$:

\enumcases{
  \item $\ell$ is attached above $x_1$ and $x_2$,
  \item $\ell$ is attached above $x_1$ and $\ell$ is below $x_2$,
  \item $\ell$ is attached above $x_2$ and $\ell$ is below $x_1$,
  \item $\ell$ is below $x_1$ and $x_2$.
}

We next treat each of these cases.
For each of them, we retrieve the pair of induced ancestors of $x_1$ and $x_2$ with the largest total weight among all pairs of ancestors induced by a label $\ell$ appropriately located with respect to $x_1$ and $x_2$.
Each of these variants gives us a candidate pair for the restricted heaviest induced ancestors of $x_1$ and $x_2$.
In the end, we return the candidate with the largest total weight.
Similar case analysis was performed in previous solutions for the HIA problem, e.g., in~\cite{DBLP:conf/cccg/GagieGN13}.

\begin{lemma}\label{le:restricted-hia-before-postprocessing}
 The answer to \restrictedhia$(v_1,v_2,B_1,B_2,x_1,x_2)$ can be retrieved from the information stored in $D[e_1,e_2]$ before the postprocessing.
\end{lemma}
\begin{proof}
By Property~\ref{prop:adding_points_on_path},
for every two points $(\pre(w_1),\pre(w_2))$ and $(\pre(w_1'),\pre(w_2'))$ added to $D[e_1,e_2]$, we have that $w_1$ is either a weak ancestor or a descendant of $w_1'$, and similarly for $w_2$ and~$w_2'$.
Hence, the preorder numbers of the nodes correspond to their depths and we can check the ancestry relation by comparing them: for nodes $u,u'$ on a path, $u$ is a weak ancestor of $u'$ if and only if $\pre(u)\leq \pre(u')$.

 Using this property, we show how to reduce each of the four cases listed above to finding a specific point in a particular rectangular subset of points.
 For now, we ignore the efficiency of the queries (a trivial implementation takes linear time) and focus on showing that the correct answer to the restricted HIA query can be retrieved from $D[e_1,e_2]$ before the postprocessing.
 
\enumcases{ 

\item
Every leaf $\ell$ that is attached above $x_1$ and $x_2$ in nodes $w_1$ and $w_2$ makes the pair $(w_1,w_2)$ a candidate result of the restricted HIA query.
Hence we need to find a point $(x,y)$ in $D[e_1,e_2]$ such that $x<\pre(x_1)$, $y< \pre(x_2)$, and $\weight(T_1[x])+\weight(T_2[y])$ is maximum.
Then, $(T_1[x],T_2[y])$ is a restricted HIA candidate pair for $(v_1,v_2)$.

\item
We need to find a point $(x,y)\in D[e_1,e_2]$ such that $x<\pre(x_1)$, $y\geq \pre(x_2)$ and $\weight(T_1[x])$ is maximized.
Then, $(T_1[x],x_2)$ is a restricted HIA candidate pair for $(v_1,v_2)$.

\item
This case is symmetric to Case 2.
We need to find a point $(x,y)\in D[e_1,e_2]$ such that $x\geq \pre(x_1)$, $y< \pre(x_2)$ and $\weight(T_2[y])$ is maximized.
Then, $(x_1,T_2[y])$ is a restricted HIA candidate pair for $(v_1,v_2)$.

\item 
We need to check if there exists a point $(x,y)$ such that $x\geq \pre(x_1)$ and $y\geq \pre(x_2)$.
If so, the pair $(x_1,x_2)$ is a restricted HIA candidate pair for $(v_1,v_2)$.
}

In each of the cases, we return a pair, if one exists, of induced ancestors of $(x_1,x_2)$ and hence also of $(v_1,v_2)$.
The label $\ell$ of leaves $\ell_1$ and $\ell_2$ that induces the pair $(r_1,r_2)$ of heaviest induced ancestors of $(v_1,v_2)$ in $B_1 \times B_2$ inserted the point to $D[e_1,e_2]$, since $B_1$ and $B_2$ are on the query paths from~$\ell_1$ and~$\ell_2$, respectively.
Hence, the pair $(r_1,r_2)$ is found while considering one of the four cases.
\end{proof}
\noindent
Now, we show that it suffices to run the above algorithm only for the points in $D[e_1,e_2]$ after the postprocessing stage.

\begin{lemma}
 The answer to \restrictedhia$(v_1,v_2,B_1,B_2,x_1,x_2)$ can be retrieved from the information stored in  $D[e_1,e_2]$ after the postprocessing.
\end{lemma}
\begin{proof}
As discussed in the proof of Lemma~\ref{le:restricted-hia-before-postprocessing}, for any two points $(\pre(w_1),\pre(w_2))$ and $(\pre(w_1'),\pre(w_2'))$ added to $D[e_1,e_2]$, $w_i$ and $w_i'$ lie on a single root-to-leaf path, so their preorder numbers are in the same order as their depths in the tree.
Recall that the trees are monotonically weighted, that is, the weights along each root-to-leaf path are increasing, so we have that $\pre(w_i)\leq \pre(w_i')$ implies $\weight(w_i)\leq \weight(w_i')$.

Let $w=(\pre(w_1),\pre(w_2))$ be the point corresponding to the answer found by the algorithm presented in Lemma~\ref{le:restricted-hia-before-postprocessing}. Note that the returned induced pair of ancestors is not necessarily $(w_1,w_2)$, e.g., it can be $(w_1,x_2)$.
Suppose that $w$ was removed during the postprocessing phase. If so, it happened because there exists a point $w'=(\pre(w_1'),\pre(w_2'))$ where $\pre(w_\star)\leq \pre(w_\star')$ and $w\ne w'$. 
If $w'$ is processed in a different case than $w$, then the pair of ancestors corresponding to $w'$ has a larger total weight than the one returned, yielding a contradiction.
If $w'$ is processed in the same case as $w$, then the pair of ancestors corresponding to $w'$ gives a pair of ancestors whose total weight is not smaller than that of the returned pair.
Hence, the reduction presented in the proof of Lemma~\ref{le:restricted-hia-before-postprocessing} still holds for the set $D[e_1,e_2]$ after the postprocessing.
\end{proof}
\noindent
Next, we present how to implement each of the four cases in Lemma~\ref{le:restricted-hia-before-postprocessing} efficiently using the fact that the points in $D[e_1,e_2]$ have been postprocessed.

\begin{lemma}\label{le:queries-after-postprocessing}
 The answer to \restrictedhia$(v_1,v_2,B_1,B_2,x_1,x_2)$ can be retrieved from the information stored in  $D[e_1,e_2]$ after the postprocessing, with two predecessor queries: one for $\pre(x_1)$ and one for $\pre(x_2)$.
\end{lemma}
\begin{proof}
By computing the predecessor of $\pre(x_1)$ in $D_x[e_1,e_2]$, we obtain intervals $\II_x^{x<\pre(x_1)}$ and $\II_x^{x\geq \pre(x_1)}$ of $D_x[e_1,e_2]$.
Similarly, intervals $\II_y^{y<\pre(x_2)}$ and $\II_y^{y\geq \pre(x_2)}$ of $D_y[e_1,e_2]$ are obtained by computing the predecessor of $\pre(x_2)$ in $D_y[e_1,e_2]$.
Note that by Property~\ref{prop:monotonicity}, points from $\II_y^{y\geq \pre(x_2)}$ (resp. $\II_y^{y< \pre(x_2)}$) of $D_y[e_1,e_2]$ correspond to points from the interval of $D_x[e_1,e_2]$ that we denote $\II_x^{y\geq \pre(x_2)}$ ($\II_x^{y<\pre(x_2)}$).
Hence, we can translate each of the conditions on points in the cases of Lemma~\ref{le:restricted-hia-before-postprocessing} to an intersection $\II_x^{\text{Case }i}$ of two intervals on $D_x[e_1,e_2]$.
This reduces each of the four cases to:

\enumcases{

\item
Find the point with maximum weight $\weight(T_1[x])+\weight(T_2[y])$ in $\II_x^{\text{Case }1}$ using an RMQ.

\item
By the monotonicity of weights with respect to $x$-coordinates, the point with maximum weight $\weight(T_1[x])$ in $\II_x^{\text{Case }2}$ is the rightmost element of $\II_x^{\text{Case }2}$.

\item
By the monotonicity of weights with respect to $y$-coordinates and Property~\ref{prop:monotonicity}, the point with maximum weight $\weight(T_2[y])$ in $\II_x^{\text{Case }3}$ is the leftmost element of~$\II_x^{\text{Case }3}$.

\item
It suffices to check if $\II_x^{\text{Case }4}$ is non-empty.\qedhere
}
\end{proof}

\begin{lemma}
 The answer to \restrictedhia$(v_1,v_2,B_1,B_2,x_1,x_2)$ can be retrieved from the information stored in $D[e_1,e_2]$ after the postprocessing, with two predecessor queries: one for $\pre(v_1)$ and one for $\pre(v_2)$.
\end{lemma}
\begin{proof}
 Recall that in the approach presented in Lemma~\ref{le:queries-after-postprocessing} we compute the predecessor of $\pre(x_1)$ in $D_x[e_1,e_2]$ in order to divide $D_x[e_1,e_2]$ into $\II_x^{x<\pre(x_1)}$ and $\II_x^{x\geq \pre(x_1)}$ and that all elements in $D_x[e_1,e_2]$ are of the form $\pre(w_1)$ for a node $w_1$ on the path from the lowest node on $e_1$ to the root of $B_1$.

We consider only $v_1$ and $x_1$ as the analysis for $v_2$ and $x_2$ is symmetric.
We can focus on the case where $v_1 \neq x_1$, as the other case is immediate.
We clearly have that $\pre(v_1)\geq \pre(x_1)$ as $v_1$ is a descendant of $x_1$.
 Recall that in the preprocessing stage, we reordered children of every node in such a way that children connected by a light edge are before those connected by a heavy edge, so if there is a child $x_1'$ of $x_1$ on $e_1$ we have $\pre(x_1) \leq \pre(v_1) < \pre(x_1')$, as $v_1$ is a descendant of a light child of $x_1$ (by the definition of $x_1$).
Hence, the predecessor of $\pre(v_1)$ is the same as the predecessor of $\pre(x_1)$ in $D_x[e_1,e_2]$.
\end{proof}

\noindent
This concludes the proof of Theorem~\ref{thm:main_data_structure}.

Finally, by setting the value of $b$ to $\lfloor \log n \rfloor$, we obtain an $\Ohtilde(n)$-space data structure capable of answering restricted HIA queries in constant time plus the time required for answering two predecessor queries: one for $\pre(v_1)$ and one for $\pre(v_2)$.
Algorithm~\ref{alg:hia} performs $\Oh(\log n/\log b)$ restricted HIA queries in order to answer an HIA query, so the total time required is $\Oh(\log n)$.
However, as all the predecessor queries ask about one of two values in different lists that are related to each other, we can make use of fractional cascading.

\subsection{Fractional Cascading}
\label{subsec:fractional_cascading}

For most of the cases described in the previous subsection, our structures
are issuing predecessor queries.
This is the only reason why the time complexity of an HIA query with our approach
is not yet $\cO(\log n / \log \log n)$.
We will exploit the fact that all these queries look for the same target value
($\pre(v_1)$ for structures built for $T_1$ and $\pre(v_2)$ for structures
for $T_2$) but for different pairs of branches, which enables us to use
fractional cascading.

We can think of creating two catalog graphs from $T_1 \times T_2$ with nodes
representing pairs $(e_1, e_2)$ of branches, storing the contents of
$D_x[e_1, e_2]$ in one catalog graph and those of $D_y[e_1, e_2]$ in the other one.
The execution of Algorithm \ref{alg:hia} can be then seen as the traversal of a path in such a
graph where, for a pair $(B_1[i],B_2[j])$ of heavy trees for which a restricted HIA query is performed
by the algorithm, we query the catalogs of the nodes representing the
pair of branches $(e_1, e_2)$ that contain the lowest weak ancestors of
$(v_1, v_2)$ that are in $B_1[i]$ and $B_2[j]$, respectively.
The problem with this direct approach is that it is not guaranteed that
the degree of all vertices in each catalog graph is polylogarithmic: we might need
to move from the node corresponding to two branches $(e_{1},e_{2})$ to any
node corresponding to two branches $(e'_{1},e'_{2})$, where the heavy tree containing $e'_{1}$
is attached to $e_{1}$, and there could be even $\Omega(n)$ such branches $e'_{1}$.

We need to create catalog graphs in which the length of the considered path for
each HIA query is
$\cO(\log n / \log \log n)$, while the degree of each node is
$\cO(\polylog n)$ in order to be able to apply the result of Shi and JáJá \cite{quingminjaja}.
This would ensure that all predecessor queries in $D_x[\cdot, \cdot]$
and $D_y[\cdot, \cdot]$ take constant time,
apart from the first ones, which take
$\cO(\log \log n)$ time using an $x$-fast trie.
We describe how to build the appropriate catalog graphs below.
As the shape of the graph for $D_x[\cdot, \cdot]$ and $D_y[\cdot, \cdot]$ is the
same and only the contents of catalogs differ, we will only describe how to build one of them.

We preprocess $T_1$ and $T_2$ separately, first to compute trees $B(T_\star)$
and then to build catalog graphs $C(T_\star)$. From this,
we build the catalog graph $C$ that can be seen as a Cartesian product of
$C(T_1)$ and $C(T_2)$.
More precisely, each node in $C$ is a pair $(v, w)$ for $v \in C(T_1)$ and
$w \in C(T_2)$.
For an edge between nodes $v_1$ and $v_2$ in $C(T_1)$,
in the final catalog graph, we create edges between nodes $(v_1, w)$ and
$(v_2, w)$ for each $w \in C(T_2)$.
Analogically, for an edge between nodes $w_1$ and $w_2$ in $C(T_2)$, in the
final catalog graph we create edges between nodes $(v, w_1)$ and $(v, w_2)$ for
each $v \in C(T_1)$.
This way, if the degrees of $C(T_1)$ and $C(T_2)$ are polylogarithmic, so is the
degree of $C$.

We now explain how to build tree
$B(T_\star)$ from $T_\star$.
$B(T_\star)$ contains nodes representing heavy trees (called
\emph{heavy tree nodes}) and nodes representing branches (called
\emph{branch nodes}):
\begin{itemize}
  \setlength{\itemsep}{1pt}
  \item for each heavy tree $H$, we connect all branch nodes representing
  branches in $H$ as children of the heavy tree node representing $H$,
  \item for each heavy tree $H$, except the tree containing the root of
  $T_\star$, we connect the heavy tree node representing $H$ as child of the
  branch node representing the branch containing the parent of the root of $H$.
\end{itemize}

\begin{proposition}
The depth of $B(T_\star)$ is $\cO(\log n / \log \log n)$ and each heavy tree
node has $\cO(\log n)$ children.
\label{prop:btree}
\end{proposition}

Recall that every heavy tree has $\cO(b)=\cO(\log n)$ branches, so every heavy tree node has $\cO(\log n)$ children.
Any root-to-leaf path in $B(T_\star)$ alternates between heavy tree nodes and
branch nodes. For any root-to-$v$ path $p$ in $T_\star$, there is a corresponding path
$p'$ in $B(T_\star)$ that visits the heavy tree nodes that correspond to the
heavy trees that intersect $p$ and the branch nodes for which Algorithm \ref{alg:hia}
(when called for a pair of nodes containing $v$)
could call predecessor queries for $D_x[\cdot, \cdot]$ and $D_y[\cdot, \cdot]$.
For any $p$, $p'$ is of length $\cO(\log n / \log \log n)$ and
can be found in $\cO(|p'|)$ time by following the path from the heavy tree containing $v$
to the root of $B(T_\star)$.

We now describe how to create a catalog graph $C(T_\star)$ from $B(T_\star)$.
The construction is recursive and follows from the proof of Lemma
\ref{lemma:gadget}.
The idea is to replace the structure of children of branch nodes having too many
children with appropriate gadgets that roughly preserve the structure of the
tree, do not increase the depth of the tree asymptotically and reduce the
degree of each node to $\cO(\polylog n)$.
Due to Proposition \ref{prop:btree}, we do not need to alter the structure of
children for heavy tree nodes.

\begin{lemma}
  For any tree $B(T_\star)$ with $\cO(n)$ nodes and depth
  $\cO(\log n / \log \log n)$,
  there is a tree $C(T_\star)$ satisfying all the following conditions:
  \begin{itemize}
    \setlength{\itemsep}{1pt}
    \item $C(T_\star)$ has $\cO(n)$ nodes,
    \item all nodes of $C(T_\star)$ have degree $\cO(\log n)$,
    \item $C(T_\star)$ has depth $\cO(\log n / \log \log n)$,
    \item for each simple path $p$ in $B(T_\star)$, there is an $\cO(\log n / \log \log n)$-length simple path $p'$ in
    $C(T_\star)$, which can be computed in $\cO(|p'|)$ time, such that $p$ is
    a subsequence of $p'$.
  \end{itemize}
  \label{lemma:gadget}
\end{lemma}
\begin{proof}
Consider a (branch) node $e$ of $B(T_\star)$ whose children, read left-to-right, are
heavy tree nodes $h_1, h_2, \ldots, h_k$ for $k>\log n$.
We replace this subgraph that contains $k+1$ nodes with a gadget graph whose root is $e$ and
whose set of nodes is a superset of $\{e\} \cup \{h_i : i \in [k]\}$.

For each node $u$, let $s(u)$ be the number of nodes in the subtree
of $u$ in the considered tree.
Let $s_0 = 0$ and,
for $i \ge 1$, let $s_i$ be the prefix sum $s(h_1) + s(h_2) + \ldots + s(h_i)$.
If there is an integer $\ell$ such that $s_{i-1} < \ell \cdot s(e) / \log n$
and $s_i \ge \ell \cdot s(e) / \log n$, we mark $h_i$.
We call each set of consecutive unmarked nodes an \emph{interval}.
As $\ell\leq \log n$, there are $\cO(\log n)$ marked nodes and $\cO(\log n)$ intervals.

We create a gadget for $e$ (and, recursively, for some other nodes created
in the construction, as described later) as follows:
\begin{itemize}
\item We attach as a child of $e$ every node that is either marked or is the only element of its interval.
\item For each interval of more than one node, we create a new node $i_j$, called an \emph{interval node}, attach it as a child of $e$, and attach all the nodes of the interval as children of $i_j$.
\end{itemize}
We recursively apply the same construction for any of the newly created interval nodes
$i_1, i_2, \ldots, i_m$ whose degree is larger than $\log n$.
See Figure~\ref{fig:gadget} for an illustration.

\begin{figure}[h]
\begin{center}
\scalebox{0.75}{\begin{tikzpicture}
  \draw[ultra thick] (0,-5) -- (0,0) node[right] {$e$};
  \draw[->] (0,-1) -- (-1,-1.5);
  \draw[->] (0,-2.5) -- (-1,-3);
  \node at (-1,-4.4) {$\ldots$};
  \draw[->] (0,-4.5) -- (-1,-5);
  \draw (-1,-1.5) -- (-1.7,-2.5) -- (-0.3,-2.5) -- (-1,-1.5);
  \draw (-1,-3) -- (-1.7,-4) -- (-0.3,-4) -- (-1,-3);
  \draw (-1,-5) -- (-1.7,-6) -- (-0.3,-6) -- (-1,-5);
  \node at (-1,-2.2) {$h_1$};
  \node at (-1,-3.7) {$h_2$};
  \node at (-1,-5.7) {$h_k$};
\end{tikzpicture}

\hspace{2cm}

\begin{tikzpicture}
  \node[draw,diamond,minimum size=0.8cm,inner sep=0cm] (e) at (0,0) {$e$};
  \node[draw,circle,minimum size=0.8cm,inner sep=0cm] (i1) at (-2,-1.5) {$i_1$};
  \node[draw,circle,minimum size=0.8cm,inner sep=0cm] (i3) at (0,-1.5) {$i_2$};
  \node (idots) at (1,-1.5) {$\ldots$};
  \node[draw,circle,minimum size=0.8cm,inner sep=0cm] (im) at (2,-1.5) {$i_m$};
  \node[draw,circle,minimum size=0.8cm,inner sep=0cm] (ii1) at (1.25,-3) {$i_{m,1}$};
  \node[draw,circle,minimum size=0.8cm,inner sep=0cm] (iimp) at (2.75,-3) {$i_{m,m'}$};
  \draw (e) -- (i1);
  \draw (e) -- (i3);
  \draw (e) -- (im);
  \draw (im) -- (ii1);
  \draw (im) -- (iimp);
  \node[draw,rectangle] (e1) at (-4,-3) {$h_1$};
  \node[draw,rectangle] (e2) at (-3.3,-3) {$h_2$};
  \node (edots) at (-2.625,-3) {$\ldots$};
  \node[draw,rectangle] (ej) at (-1.95,-3) {$h_{j_1}$};
  \node[draw,rectangle] (ejp1) at (-0.95,-3) {$h_{{j_1}+1}$};
  \node (edots2) at (0,-3) {$\ldots$};
  \node[draw,rectangle] (ejjp1) at (1,-4.5) {$h_{j_{m-1}+1}$};
  \node (edots3a) at (1,-3.75) {$\ldots$};
  \node (edots3b) at (2,-3.75) {$\ldots$};
  \node (edots3c) at (3,-3.75) {$\ldots$};
  \node (edots3d) at (2,-3) {$\ldots$};
  \node (edots3e) at (2.25,-4.5) {$\ldots$};
  \node[draw,rectangle] (ek) at (3,-4.5) {$h_k$};
  \draw (i1) -- (e1);
  \draw (i1) -- (e2);
  \draw (i1) -- (ej);
  \draw (e) -- (ejp1);
  \draw (ii1) -- (edots3a);
  \draw (iimp) -- (edots3c);
  \draw (edots3a) -- (ejjp1);
  \draw (edots3c) -- (ek);
\end{tikzpicture}}
\caption{On the left there is a branch $e$ of $T_\star$, with outgoing edges to
$h_1, h_2, \ldots, h_k$ in $B(T_\star)$. On the right, there is a catalog graph
gadget created for $e$, which is part of $C(T_\star)$.
Circles denote interval nodes and rectangles heavy tree nodes.
Some of the intervals are recursively replaced	 with the gadget to decrease their degree.
Heavy tree nodes have degree $\cO(\log n)$.
In $C(T_\star)$, $e$'s parent is the heavy tree to which it belongs, while the children of
each $h_i$ are the branches in $h_i$.}
\label{fig:gadget}
\end{center}
\end{figure}
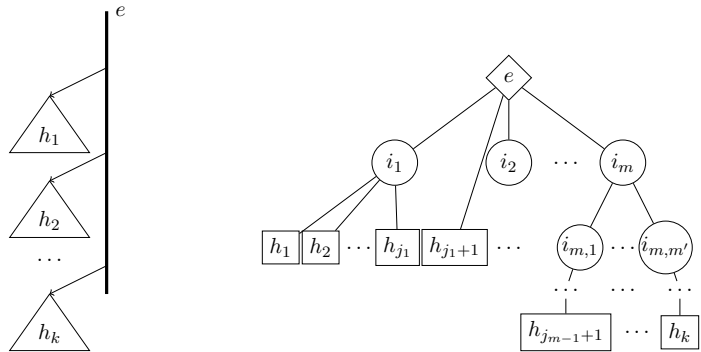

From the construction, it follows that the degree of each node of
$C(T_\star)$ is $\cO(\log n)$ and that the size of $C(T_\star)$ is $\cO(n)$,
as all new nodes are of out-degree at least $2$.

Let $u$ be a node in $T_\star$.
We now show that the depth for a node $e_u \in C(T_\star)$ representing a branch
containing $u$ is $\cO(\log n / \log \log n)$ by considering the edges above $e_u$ in $C(T_\star)$.
The edges can be of two types:
\begin{itemize}
  \setlength{\itemsep}{1pt}
  \item Edges that are incident to at least one node that is not an interval node. By
  the depth of $B(T_\star)$, we have $\cO(\log n / \log \log n)$ such edges.
  \item Edges between interval nodes. For each such edge $(v,z)$,
   we have $s(v) \ge s(z) \cdot \log n$. Thus, similarly to the proof
  of Lemma \ref{lemma:layers_count}, there are
  $\cO(\log n / \log \log n)$ such edges on the path from the root of $C(T_\star)$ to $e_u$.
\end{itemize}
This concludes the proof of the bound on the depth of $C(T_\star)$.

Each simple path $p$ in $B(T_\star)$ naturally corresponds to a
simple path $p'$ in $C(T_\star)$.
In partcular, for each edge $(v,w)$ in $B(T_\star)$, one can explicitly store a path
$C(T_\star)$ to which $(v,w)$ corresponds.
The concatenation of all such paths for edges on $p$ yields a simple path $p'$ in $\cO(|p'|)$ time.
As the depth of $C(T_\star)$ is $\cO(\log n/ \log\log n)$, the bound on $|p'|$ follows.
\end{proof}

From $C(T_1)$ and $C(T_2)$ constructed as in Lemma \ref{lemma:gadget}, we create
the catalog graph $C$ as described earlier.
Only nodes that represent pairs of branches contain non-empty original catalogs.
After the original catalogs are filled, we run the preprocessing of fractional
cascading and appropriate augmented catalogs are created for all nodes in $C$ as
described in \cite{chazelle,quingminjaja}.
The $\cO(\log n / \log \log n)$ predecessor queries coming from restricted HIA queries
performed in Algorithm \ref{alg:hia} are naturally reduced to a constant
number of queries to the $x$-fast tries and the traversal of a path of length
$\cO(\log n / \log \log n)$ in $C$.
This takes $\cO(\log n / \log \log n)$ time in total and concludes the description of our data structure and the proof of Theorem~\ref{thm:final}.

\bibliographystyle{plainurl}
\bibliography{references}

\newpage
\appendix

\section{Description of the $\cO(\log n)$-Query-Time Data Structure of \cite{DBLP:conf/cccg/GagieGN13}}\label{app:fc}

As mentioned in the introduction, Gagie et al.~\cite{DBLP:conf/cccg/GagieGN13} sketched an
$\cO(n \log^2 n)$-size data structure that answers HIA queries in $\cO(\log n)$ time,
in the last paragraph of Subsection 2.1 of their work.
They construct a data structure for each pair of heavy trees of $T_1$ and $T_2$ and reduce an HIA query for nodes $u$ and $v$ to a predecessor query in the data structure of each of $\cO(\log n)$ pairs of heavy trees.
The idea for improving the (fully-described) $\cO(\log n \log\log n)$-time procedure for answering queries with a more efficient one is similar to ours and involves fractional cascading.
They would need to build a catalog graph with nodes being pairs of heavy paths, reduce its degree,
and make sure that the predecessor queries have the same target (by asking for the preorder numbers of $u$ and $v$).
This is similar to what we describe in \cref{subsec:fractional_cascading} for a different tree decomposition.

\end{document}